\documentclass[11pt]{article}
\usepackage{preamble}

\title{Distribution-Free Testing of Linear Functions on $\mathbb{R}^n$}
\author{Noah Fleming, Yuichi Yoshida}
\begin{document}
\maketitle

\begin{abstract}
We study the problem of testing whether a function $f\colon \mathbb{R}^n \rightarrow \mathbb{R}$ is  linear (i.e., both additive and homogeneous) in the \emph{distribution-free} property testing model, where the distance between functions is measured with respect to an unknown probability distribution over $\mathbb{R}$. We show that, given query access to $f$, sampling access to the unknown distribution as well as the standard Gaussian, and $\varepsilon > 0$, we can distinguish additive functions from functions that are $\varepsilon$-far from additive functions with $O\left(\frac{1}{\varepsilon}\log\frac{1}{\varepsilon}\right)$ queries, independent of $n$. Furthermore, under the assumption that $f$ is a continuous function, the additivity tester can be extended to a distribution-free tester for linearity using the same number of queries. On the other hand, we show that if we are only allowed to get values of $f$ on sampled points, then any distribution-free tester requires $\Omega(n)$ samples, even if the underlying distribution is the standard Gaussian.
\end{abstract}

\thispagestyle{empty}
\setcounter{page}{0}
\newpage


\section{Introduction}

Property testing of Boolean functions is a well studied area in which, given query access to a function $f\colon {\{0,1\}}^n \to \{0,1\}$ and a parameter $\varepsilon > 0$, the goal is to distinguish with high probability the case that $f$ satisfies some predetermined property $P$ from the case that $f$ is $\varepsilon$-far from satisfying $P$, that is, we need to change the values of $f(x)$ for at least an $\varepsilon$-fraction of $x \in {\{0,1\}}^n$ to make $f$ satisfy $P$.
Since the seminal work by Blum, Luby and Rubinfeld~\cite{BlumLR93}, property testing has become a thriving field, and many properties of Boolean functions have been shown to be  testable with a number of queries independent of $n$, including linear functions~\cite{BlumLR93}, low-degree polynomials~\cite{KaufmanR06,Bhattacharyya:2010kw} and $k$-juntas~\cite{Fischer:2004ds,Blais:2009kt,Blais:2015dt}.
For an introductory survey, we recommend~\cite{Goldreich17}.




In contrast to Boolean functions, only a few properties of functions on a Euclidean space, that is, $\mathbb{R}^n$, have been studied.
For a measurable function $f\colon \mathbb{R}^n \to \mathbb{R}$, $\varepsilon > 0$, and a property $P$, we say that $f$ is \emph{$\varepsilon$-far} from $P$ if
\[
  \Pr_{x \sim \mathcal{N}(0,I)}[f(x) \neq g(x)]>\varepsilon,
\]
for any measurable function $g\colon \mathbb{R}^n \to \mathbb{R}$ satisfying $P$, where $\mathcal{N}(0,I)$ is the standard Gaussian.
We say that an algorithm is a \emph{tester} for a property $P$ if, given query access to a measurable function $f\colon \mathbb{R}^n \to \mathbb{R}$, sampling access to the standard Gaussian, and $\varepsilon > 0$, it accepts with probability at least $2/3$ when $f$ satisfies $P$, and rejects with probability at least $2/3$ when $f$ is $\varepsilon$-far from $P$.
Testability of a variety of properties has been considered, including surface area of a set~\cite{KothariNOW14, Neeman14}, half spaces~\cite{Matulef:2010ds, matulef2009testing, MatulefORS10}, linear separators~\cite{BalcanBBY12}, high-dimensional convexity~\cite{0001FSS17}, and linear $k$-junta~\cite{2018arXiv180610057D}.


Although the standard Gaussian is natural, it barely appears in practice. In fact, we typically have little, if any, information about the underlying distribution. This raises the question of whether we can test when the underlying distribution of the data is unknown.
For a measurable function $f\colon \mathbb{R}^n \to \mathbb{R}$, $\varepsilon > 0$, a distribution $\mathcal{D}$ over $\mathbb{R}^n$, and a property $P$, we say that $f$ is \emph{$\varepsilon$-far} from $P$ with respect to $\mathcal{D}$ if
\[
  \Pr_{x \sim \mathcal{D}}[f(x) \neq g(x)]>\varepsilon,
\]
for any measurable function $g\colon \mathbb{R}^n \to \mathbb{R}$ satisfying $P$.
We say that an algorithm is a \emph{distribution-free tester} for a property $P$ if, given query access to a measurable function $f\colon \mathbb{R}^n \to \mathbb{R}$, sampling access to an \emph{unknown} distribution $\mathcal{D}$ over $\mathbb{R}^n$ as well as the standard Gaussian, and $\varepsilon > 0$, it accepts with probability at least $2/3$ when $f$ satisfies $P$, and rejects with probability at least $2/3$ when $f$ is $\varepsilon$-far from $P$ with respect to $\mathcal{D}$.
Distribution-free property testing is an attractive model because it makes minimal assumptions on the environment, and models the scenario most often occurring in practice.

We say that a function $f\colon \mathbb{R}^n \to \mathbb{R}$ is \emph{additive} if $f(x)+f(y) = f(x+y)$ for any $x,y \in \mathbb{R}^n$.
In this work, we consider distribution-free testing of additivity of functions $f\colon \mathbb{R}^n \to \mathbb{R}$ and show the following.
%
%
\begin{thm}\label{thm:intro-additivity}
  There exists a one-sided error distribution-free tester for additivity of $f\colon\mathbb{R}^n \rightarrow \mathbb{R}$ with $O\left(\frac{1}{\varepsilon} \log \frac{1}{\varepsilon}\right)$ queries.
\end{thm}
Previously no algorithm was known even when the underlying distribution $\mathcal{D}$ is the standard Gaussian.
As there is a trivial lower bound of $\Omega\left(\frac{1}{\varepsilon}\right)$, the query complexity of our tester is almost tight.

We say that a function $f\colon \mathbb{R}^n \to \mathbb{R}$ is \emph{homogeneous} if $cf(x) = f(cx)$ for any $x \in \mathbb{R}^n$ and $c \in \mathbb{R}$.
A function that is both additive and homogeneous is said to be \emph{linear}.
Although additivity and linearity are equivalent for functions over finite groups, there are (pathological) functions $f\colon \mathbb{R}^n \to \mathbb{R}$ that are additive but not homogeneous.
Hence, the testability of additivity does not immediately imply the testability of linearity.
However, when the input function is guaranteed to be continuous, we can also test linearity.
\begin{thm}\label{thm:intro-linearity}
  Suppose that the input function is guaranteed to be continuous.
  Then, there exists a one-sided error distribution-free tester for linearity with $O\left(\frac{1}{\varepsilon} \log \frac{1}{\varepsilon}\right)$ queries.
\end{thm}


It is also natural to assume that we can get values of the input function only on sampled points.
Specifically, we say that a (distribution-free) tester is \emph{sample-based} if it accesses the input function $f\colon \mathbb{R}^n \to \mathbb{R}$ through points sampled from the distributions $\mathcal{D}$ and $\mathcal{N}(0,I)$.
We show a strong lower bound for sample-based testers.
\begin{thm}\label{thm:intro-lower-bound}
  Any sample-based tester for the linearity of functions $f\colon \mathbb{R}^n \to \mathbb{R}$ requires $\Omega(n)$ samples, even when $\mathcal{D} = \mathcal{N}(0,I)$.
\end{thm}
This theorem shows a sharp contrast between query-based and sample-based testers in their testability of properties of functions on a Euclidean space.
We note that we can show the same lower bound for testing additivity with an almost identical proof.

\subsection{Related Work}
The question of property testing first appeared (implicitly) in the work of Blum, Luby and Rubinfeld~\cite{BlumLR93}. Among the problems that they studied was linearity testing. Their algorithm, now famously known as the BLR test, has played a key role in the design of probabilistically checkable proofs~\cite{AroraLMSS92, BellareGS95, Hastad96} and this connection was some of the early motivation for the field of property testing.
Since the original paper, the parameters of the BLR test have been extensively refined. Much of this work focused on reducing the amount of randomness, due to this being a key parameter in probabilistically checkable proofs, as well as analyzing the rejection probability (see~\cite{raskhodnikova2014linearity} for a survey). Another line of works considered the testing linearity over more general domains. The works of~\cite{BlumLR93, BenOr:2008cg,Oono:2016ee} showed that the BLR test can be used to test the linearity of any function with $f\colon G \rightarrow H$ for finite groups $G$ and $H$ with $O(1/\varepsilon)$ queries. Following this, a body of work~\cite{GemmellLRSW91,ErgunKR01,ArBCG93,KiwiMS99} constructed testers for linearity of functions $f\colon S \rightarrow \mathbb{R}$, where $S$ is a finite subset of rational numbers, and the distance is measured with respect to the uniform distribution over $S$. See~\cite{KiwiMS00} for a survey. These results were phrased in terms of approximate self-testing and correcting programs. In this setting the queries to $f$ return a finite approximation of $f(x)$. Although these results are arguably the most related to our work, our proof differs significantly from theirs and instead takes inspiration from the original BLR test.



Distribution-free testing (for graph properties) was first defined by Goldreich~et~al.~\cite{GoldreichGR98}, though the first distribution-free testers for non-trivial properties appeared much later in the work of Halevy and Kushilevitz~\cite{HalevyK07}. Subsequently, distribution-free testers have been considered for a variety of Boolean functions including low-degree polynomials, dictators, and monotone functions~\cite{HalevyK07}, $k$-juntas~\cite{HalevyK07,LiuCSSX19,Bshouty19,Belovs19}, conjunctions, decision lists, and linear threshold functions~\cite{GlasnerS07}, monotone and non-monotone monomials~\cite{DolevR11}, and monotone conjunctions~\cite{GlasnerS07, ChenX16}. However, to our knowledge the only (partial) distribution-free tester for a class of function on the Euclidean space is due to Harms~\cite{Harms19} who gave an efficient tester for half spaces, that is, functions $f\colon \mathbb{R}^n \to \{0,1\}$ of the form $f(x) = \sgn(w^\top x -\theta)$ for some $w \in \mathbb{R}^n$ and $\theta \in \mathbb{R}$, over any rotationally invariant distribution.


\subsection{Proof Technique}

The construction of our tester for additivity will be done in two steps. First, we construct a constant-query tester for additivity over the standard Gaussian distribution $\mathcal{N}(0,I)$. Our tester will accept linear functions with probability $1$, and so the majority of the work is in showing that if the test accepts the given function $f\colon \mathbb{R}^n \to \mathbb{R}$ then, with high probability, $f$ is close to an additive function. To do so, we show that if $f$ passes a series of tests then there exists a related function $g\colon\mathbb{R}^n \to\mathbb{R}$, defined from $f$, which is additive. Furthermore, if $f$ is linear then $f=g$. The definition of $g$ will allow us to obtain query access to it with high probability, and so we can simply estimate the distance between $f$ and $g$. At a high-level, this is somewhat similar to the BLR test, however operating over $\mathcal{N}(0,I)$ rather than the uniform distribution presents its own set of non-trivial challenges. We discuss these, as well as the definition of $g$ at the start of Section~\ref{sec:additivity-gaussian}.

It is fairly straightforward to generalize this tester for additivity to a distribution-free tester. To do so, we run the additivity tester for the standard Gaussian, except that testing the distance between $f$ and $g$ will now be done using samples from the unknown $\mathcal{D}$. This crucially relies on our ability to draw samples from the standard Gaussian.

Any additive function $f\colon \mathbb{R}^n \rightarrow \mathbb{R}$ is linear over the rationals, meaning that $f(qx) = qf(x)$ for every $q \in \mathbb{Q}$. Therefore, in order to test linearity it remains to test whether this holds also for irrationals. Assuming that $f$ is continuous we are able to modify our tester to show that this implies that the additive function $g$ is continuous as well. We then leverage the fact that any continuous additive function is linear in order to obtain our linearity tester.

To prove Theorem~\ref{thm:intro-lower-bound},
the lower bound on sample-based testers for linearity, we construct two distributions, one supported on linear functions, and the other supported on functions which are far from linear. Consider drawing a function $f$ from one of these two distributions with equal probability.
By Yao's minimax principle it suffices to show that any deterministic algorithm which receives $n$ samples from $\mathcal{N}(0,I)$, together with their evaluations on $f$, is unable to distinguish, with high probability, which of the two distribution $f$ came from. To construct the distribution on linear functions, we sample $w \sim \mathcal{N}(0,I)$ and return $f(x):= w^\top x$. Our distribution on functions which are far from linear is designed so that any function $f$ from this distribution satisfies $f(x+y) \neq f(x) + f(y)$ with probability $1$ over $x,y \sim \mathcal{N}(0,I)$. To do so, for every $x \in \mathbb{R}^n$ we sample $\varepsilon_x$ from a one-dimensional Gaussian and return $f(x):= w^\top x +\varepsilon_x$. It is not difficult to show that such functions are far from linear. 

\subsection{Organization}
The remainder of the paper is organized as follows. We reivew basic facts on probability distributions in Section~\ref{sec:pre}. In Section~\ref{sec:additivity} we develop our distribution-free tester for additivity by first constructing a tester for additivity over the standard Gaussian in Section~\ref{sec:additivity-gaussian}. We generalize this tester to the distribution-free setting in Section~\ref{sec:distribution-free-additivity} and to a tester for linearity in Section~\ref{sec:linearity}. Finally, we end with our lower bound on the sampling model in Section~\ref{sec:lower-bound}.


\section{Preliminaries}\label{sec:pre}

Let $\cal D$ and $\cal D'$ be probability distributions on the same domain $\Omega$. Then, the \emph{total variation distance} between them, denoted by $\dtv(\cal D, D')$, is defined as
\[ \dtv(\mathcal{D}, \mathcal{D}'):= \frac{1}{2} \int_{\Omega} |\mathcal{D}(x) - \mathcal{D}'(x)| dx. \]
The \emph{Kullback-Leibler divergence} (or \emph{KL-divergence}) of $\mathcal{D}'$ from $\cal D$, denoted $\dkl(\mathcal{D} \| \mathcal{D}')$, is defined as
\[ \dkl(\mathcal{D} \| \mathcal{D}') = \int_{\Omega} \mathcal{D}(x) \log \left( \frac{\mathcal{D}(x)}{\mathcal{D}'(x)} \right) dx. \]
We will use the KL-divergence to upper bound the total variation distance, using the following inequality. 
\begin{thm}[Pinsker's Inequality]\label{thm:Pinskers} Let $\mathcal{D}$ and $\mathcal{D}'$ be probability distributions on the same domain $\Omega$. Then,
  \[ \dtv(\mathcal{D}, \mathcal{D}') \leq \sqrt{ \frac{1}{2} \dkl (\mathcal{D} \| \mathcal{D}')}.\]
\end{thm}
The following allows us to bound the KL-divergence between two Gaussian distributions.
\begin{lem}\label{lem:TVD-bound-on-Gaussian}
Let $\mathcal{D} = \mathcal{N}(\mu_1, \Sigma_1)$ and  $\mathcal{D}' = \mathcal{N}(\mu_2, \Sigma_2)$ be multivariate Gaussian distributions with $\mu_1, \mu_2 \in \mathbb{R}^n$ and invertible $\Sigma_1, \Sigma_2 \in \mathbb{R}^{n \times n}$. Then,
\[ \dkl(\mathcal{D} \| \mathcal{D}') = \frac{1}{2} \left( \log  \left(\frac{\det \Sigma_2}{ \det \Sigma_1} \right) + \trace \Big( {(\Sigma_2)}^{-1} \Sigma_1 \Big) - n + {(\mu_2 - \mu_1)}^\top \Sigma_2^{-1} (\mu_2 - \mu_1)\right). \]
\end{lem}

We record a useful lemma about total variation distance of Gaussians with shared covariance matrices.
\begin{lem}\label{lem:bound-TV-different-means}
  Consider two Gaussian distributions $\mathcal{N}(\mu_1,\Sigma), \mathcal{N}(\mu_2,\Sigma)$ with shared invertible covariance matrices $\Sigma \in \mathbb{R}^{n \times n}$. Then $\dtv(\mathcal{N}(\mu_1,\Sigma), \mathcal{N}(\mu_2,\Sigma)) \leq \phi$ holds if $\| \mu_1 - \mu_2 \|_2 \leq  2 \phi / \sqrt{ \| \Sigma^{-1} \|_2}$.
\end{lem}
\begin{proof}
  Denote  $\mu := \mu_1 - \mu_2$.
  Applying Lemma~\ref{lem:TVD-bound-on-Gaussian}, we have $
    \dtv(\mathcal{N}(\mu_1,\Sigma), \mathcal{N}(\mu_2, \Sigma)) = \sqrt{\frac{1}{4}\mu^\top \Sigma^{-1} \mu}$. Now, because $\Sigma$ is PSD, $\mu^\top \Sigma^{-1} \mu \leq \|\mu\|_2^2 \|\Sigma^{-1}\|_2$, where $\| \cdot \|_2$ is the spectral matrix norm. Therefore, we have  $\dtv(\mathcal{N}(\mu_1,\Sigma), \mathcal{N}(\mu_2, \Sigma)) \leq  \frac{1}{2} \|\mu\|_2 \sqrt{  \|\Sigma^{-1}\|_2} \leq \phi$.


\end{proof}

\section{Testing Additivity}\label{sec:additivity}
In this section, we develop our distribution-free tester for additivity.
For convenience, we first describe a simpler tester for additivity over the standard Gaussian distribution $\mathcal{N}(0, I)$ in Section~\ref{sec:additivity-gaussian}.
Then, in Section~\ref{sec:distribution-free-additivity}, we describe how to generalize this algorithm to test additivity over an unknown distribution.

\subsection{Tester for the Standard Gaussian}\label{sec:additivity-gaussian}

Our goal in this section is to design a constant-query tester for the additivity of a measurable function $f\colon \mathbb{R}^n \rightarrow \mathbb{R}$ over the standard Gaussian.

\begin{thm}\label{thm:central-gausssian}
  There exists a one-sided error $\Omega\left(\frac{1}{\varepsilon} \log\frac{1}{\varepsilon}\right)$-query tester for additivity over the standard Gaussian.
\end{thm}

At a high-level, our tester consists of two steps. First, we test whether $f$ satisfies additivity over a set of samples drawn from the distribution. If $f$ passes this tests, then we conclude that there must be an additive function $g \colon\mathbb{R}^n \rightarrow \mathbb{R}$, which is a self-corrected version of $f$. Second, by testing the value of $f$ on a correlated set of points, we are able to get query access to $g$ with high probability, and therefore we can simply estimate the distance between $f$ and $g$. Our tester relies on the fact that it has one-sided error: if $f$ is additive then our test passes with probability $1$. Otherwise, if $f$ is non-additive and the second step passes, then with high probability $f$ and $g$ must be close.

The first step is inspired by the BLR test. Indeed, the evaluation of the function $g$ at a point $p$ is defined as the (weighted) majority value of $f(p-x) + f(x)$ over all $x \sim \mathcal{N}(0,I)$ (where, $f(p-x) + f(x)$ is weighted according to the probability of drawing $x \sim \mathcal{N}(0,I)$). However, there are some significant challenges in generalizing the BLR test to the standard Gaussian, the most obvious of which is that unlike the uniform distribution, every point in the support of the distribution does not have equal probability. In particular, $p-x$ is not distributed as $x \sim \mathcal{N}(0,I)$ for fixed $p \neq 0$. In order to overcome this, we exploit the fact that for additive functions $f$, we have $f(x) = q f(x/q)$ for every rational $q$. This allows us to restrict attention to a small ball $\ball(0,1/r)$ of radius $1/r$ centred at the origin. Then, for $p \in \ball(0,1/r)$, $p-x$ is approximately distributed as $x$ for small enough $1/r$. Thus, we get around the issue of unevenly weighted points by defining $g$ within $\ball(0,1/r)$, and then extrapolating to define $g$ over $\mathbb{R}^n$.

Concretely, we will define $g$ as follows. First, let $r$ be a sufficiently large integer ($r=50$ suffices). For each point $p \in \mathbb{R}^n$ define
\[ k_p := \begin{cases} 1 &\text{ if } \|p \|_2 \leq 1/r, \\ \left \lceil r \cdot \|p\|_2 \right \rceil   &\text{ if } \|p \|_2 > 1/r. \end{cases}\]
Now, define $g\colon\mathbb{R}^n \rightarrow \mathbb{R}$ as
\[ g(p) := k_p \cdot \mathop{\maj}_{\mathcal{N}(0,I)} \left[  f \left( \frac{p}{k_p}-x \right) +f \left (x \right)  \right], \]
where $\maj_{\mathcal{N}(0,I)}$ is the \emph{weighted majority function} where a value $f(p/k_p - x) + f(x)$ is weighted according to its probability mass under $x \sim \mathcal{N}(0,I)$. Observe that either $p \in \ball(0,1/r)$, or $g(p)$ first maps $p$ to a point $p/k_p$ in $\ball(0,1/r)$. The value of $g$ is the most likely value (according to $\mathcal{N}(0,I)$) of  $f(p/k_p -x) + f(x)$. If $f$ is close to additive, then taking this majority should allow us to correct for the errors in $f$.

An equivalent definition of $g$ which will be useful is the following. For $p \in \mathbb{R}^n$ let $P_p$ be the Lebesgue measurable function such that $\int_A P_p(x) dx$ gives the probability (over $\mathcal{N}(0,I)$) that $f(p/k_p-x) + f(x)$ takes value in $A$. Then $g$ is defined as $g(p) := \argmax_x P_p(x)$ if $P_p(x) \geq 1/2$.



\vspace{1em}

\begin{algorithm}[!t]
  \caption{Standard Gaussian Additivity Tester}\label{alg:zero-mean-additivity}
    \Given{Query access to $f\colon \mathbb{R}^n \rightarrow \mathbb{R}$, sampling access to the distribution $\mathcal{N}(0, I)$;}
    \textbf{Reject} if \Call{TestAdditivity}{$f$} returns \textbf{Reject}\;
    \For{$N_{\ref{alg:zero-mean-additivity}} := O(1/\varepsilon)$ times}{
      Sample $p \sim \mathcal{N}(0,I)$\;
      \textbf{Reject} if $f(p) \neq$ \Call{Query-$g$}{$p,f$} or if \Call{Query-$g$}{$p,f$} returns \textbf{Reject}.
    }
    \textbf{Accept}.
\end{algorithm}
\begin{algorithm}[!t]
  \caption{Subroutines}\label{alg:subroutines}
  \Procedure{\emph{\Call{TestAdditivity}{$f$}}}{
    \Given{Query access to $f\colon \mathbb{R}^n \rightarrow \mathbb{R}$, sampling access to the distribution $\mathcal{N}(0, I)$;}
    \For{$N_{\ref{alg:subroutines}} := O(1)$ times}{
      Sample $x, y, z \sim \mathcal{N}(0,I)$\;
      \textbf{Reject} if $f(-x) \neq -f(x)$\;
      \textbf{Reject} if $f(x - y) \neq f(x) - f(y)$\;
      \textbf{Reject} if $f\left(\frac{x-y}{2} \right) \neq f \left(\frac{x - z}{2} \right) + f \left(\frac{z - y}{2} \right)$\;
    }
    \textbf{Accept}.
  }
  \Procedure{\emph{\Call{Query-$g$}{$p, f$}}}{
    \Given{$p \in \mathbb{R}^n$, query access to $f\colon \mathbb{R}^n \rightarrow \mathbb{R}$, sampling access to $\mathcal{N}(0,I)$;}
    $N'_{\ref{alg:subroutines}} := O(\log \frac{1}{\varepsilon})$\;
    Sample $x_{1}, \ldots, x_{N'_{\ref{alg:subroutines}}} \sim \mathcal{N}(0,I)$\;
    \textbf{Reject} if there exists $i,j \in [ N'_{\ref{alg:subroutines}}]$ such that $f(p/k_p - x_i) + f(x_i) \neq f(p/k_p - x_j) + f(x_j)$\;
    \Return $k_p \left( f( p/k_p-x_1) + f(x_1) \right)$.
  }
\end{algorithm}

Our algorithm is given in Algorithm~\ref{alg:zero-mean-additivity}, which uses subroutines given in Algorithm~\ref{alg:subroutines}.
The \Call{Query-$g$}{} subroutine allows us to obtain query access to $g$ with high probability, while the \Call{TestAdditivity}{} subroutine tests the conditions that we require in order to prove that $g$ is additive.

\begin{lem}\label{lem:g-is-additive}
  If \Call{TestAdditivity}{$f$} accepts with probability at least $1/10$, then $g$ is a well-defined, additive function, and furthermore, $\Pr_{x \sim \mathcal{N}(0,I)}[g(x) \neq k_p (f(p/k_p - x) + f(x)) ] < 1/2$.
\end{lem}

We first prove Theorem~\ref{thm:central-gausssian} assuming that Lemma~\ref{lem:g-is-additive} holds.
\begin{proof}[Proof of Theorem~\ref{thm:central-gausssian}]
  First, observe that if $f$ is an additive function then Algorithm~\ref{alg:zero-mean-additivity} always accepts. It is immediate that \Call{TestAdditivity}{$f$} always accepts. To see that it also passes the remaining tests, observe that by additivity,  $k_p \left( f(p/k_p - x) + f(x) \right) = k_p  f(p/k_p ) = f(p)$, where the final inequality holds because $k_p \in \mathbb{Z}$ and by homogeneity over the rationals $f(qx) = qf(x)$ for every $q \in \mathbb{Q}$.

  We now show that if $f$ is $\varepsilon$-far from additive functions, then Algorithm~\ref{alg:zero-mean-additivity} rejects with probability at least $2/3$.
  If \Call{TestAdditivity}{$f$} accepts with probability at most $1/10$, we can reject $f$ with probability at least $1-1/10>2/3$.
  Hence, we assume that \Call{TestAdditivity}{$f$} accepts with probability at least $1/10$.
  Then by Lemma~\ref{lem:g-is-additive}, the function $g$ is additive and hence $f$ is $\varepsilon$-far from $g$.
  Now, we want to bound the probability that Step~3 of Algorithm~\ref{alg:zero-mean-additivity} passes.

  First, we bound the probability that \Call{Query-$g$}{$p, f$} fails to recover the value of $g(p)$. That is, we bound the probability that $f(p/k_p -x_i) + f(x_i) = f(p/k_p -x_j) + f(x_j)$ for all $i,j \in \left[ N'_{\ref{alg:subroutines}} \right]$, but $g(p) \neq k_p \left( f(p/k_p-x_i) + f(x_i) \right)$.
  By Lemma~\ref{lem:g-is-additive}, the probability that we draw $N'_{\ref{alg:subroutines}}$ points which satisfy this is at most $2^{-N'_{\ref{alg:subroutines}}} \leq \varepsilon /2$ by choosing the hidden constant in $N'_{\ref{alg:subroutines}}$ to be large enough. Therefore the probability that we correctly recover $g(p)$ is at least $1-\varepsilon/2$.

  Now that we have established that we can obtain query access to $g$ with high probability, it remains to show that we can test whether $f$ and $g$ are close.
  Indeed, the probability that Step 3 of Algorithm~\ref{alg:zero-mean-additivity} fails to reject is at most
  \begin{align*}
    & {\left( \Pr_{p \sim \mathcal{N}(0,I)} \left[ f(p) = g(p) \vee \text{\Call{Query-$g$}{$p, f$} fails to correctly recover $g(p)$} \right]  \right)}^{N_{\ref{alg:zero-mean-additivity}}} \\
    \leq & {\left( 1- \Pr_{p \sim \mathcal{N}(0,I)} [f(p) \neq g(p)] + \Pr_{p \sim \mathcal{N}(0,I)} \left[ \text{\Call{Query-$g$}{$p, f$} fails to correctly recover $g(p)$} \right]  \right)}^{N_{\ref{alg:zero-mean-additivity}}} \\
    <& {\left(1- \frac{\varepsilon}{2} \right)}^{N_{\ref{alg:zero-mean-additivity}}} < \frac{1}{10},
  \end{align*}
  by choosing the hidden constant in $N_{\ref{alg:zero-mean-additivity}}$ to be large enough.
  Therefore, Algorithm~\ref{alg:zero-mean-additivity} rejects with probability at least $1-1/10 > 2/3$.
\end{proof}

It remains to prove Lemma~\ref{lem:g-is-additive} showing that if  Algorithm~\ref{alg:zero-mean-additivity} succeeds, then $g$ is an additive function with high probability.

\subsubsection{Additivity of the Function $g$}
First, we record the basic, but useful observation that if the \Call{TestAdditivity}{} subroutine passes then each of its tests hold with high probability over $\mathcal{N}(0,I)$.

\begin{lem}\label{lem:f-additive-whp}
  If \Call{TestAdditivity}{$f$} accepts with probability at least $1/10$, then
  \begin{align} \Pr_{x,y \sim \mathcal{N}(0,I)} \left [f(x - y) = f(x) - f(y) \right] \geq \frac{99}{100}, \label{eq:lem3-1} \\
  \Pr_{x \sim \mathcal{N}(0,I)}[f(-x) = -f(x)] \geq \frac{99}{100}, \label{eq:lem3-2} \\
  \Pr_{x,y,z \sim \mathcal{N}(0,I)} \left [f\left(\frac{x-y}{2} \right) = f \left(\frac{x - z}{2} \right) + f \left(\frac{z - y}{2} \right) \right] \geq \frac{99}{100}.  \label{eq:lem3-3}
  \end{align}
\end{lem}
\begin{proof}
  Suppose for contradiction that at least one of~\eqref{eq:lem3-1},~\eqref{eq:lem3-2}, and~\eqref{eq:lem3-3} does not hold.
  We here assume that~\eqref{eq:lem3-1} does not hold as other cases are similar.

  We accept only when all the sampled pairs $(x,y)$ satisfy $f(x+y) = f(x) + f(y)$.
  By setting the hidden constant in $N_{\ref{alg:subroutines}}$ to be large enough, this happens with probability at most
  \[ {\left(1- \Pr_{x,y \sim \mathcal{N}(0,I)} [f(x+y) \neq f(x) + f(y)] \right)}^{N_{\ref{alg:subroutines}}} < {\left( \frac{99}{100}\right)}^{N_{\ref{alg:subroutines}}} < \frac{1}{10}, \]
  which is a contradiction.
\end{proof}

In order to argue that $g$ is additive, we will first argue that $g$ is additive on points within the tiny ball $
\ball(0,1/r)$. To do so, we will crucially use the fact that $p-x$ is distributed approximately as $x \sim \mathcal{N}(0,I)$ if $\|p\|_2$ is small.
By Lemma~\ref{lem:bound-TV-different-means} we have a bound on the total variation distance between $x$ and $x+p$.
\begin{claim}\label{clm:tvd-bound-on-ball}
  Let $p \in \mathbb{R}^n$ satisfying $\|p\|_2 \leq k/r$ for some $k \in \mathbb{Z}^{>0}$.
  Then $\dtv(\mathcal{N}(0,I), \mathcal{N}(p, I)) \leq k /100$.
\end{claim}
\begin{proof}
	By Lemma~\ref{lem:bound-TV-different-means}, for $\dtv(\mathcal{N}(0,I), \mathcal{N}(p, I)) \leq k /100$ it is enough that $p$ satisfies $\|0-p\|_2 \leq 2k/(100\sqrt{\|I\|_2})$. Because $\| p\|_2 \leq k/r = 2k/100 = 2k/(100\sqrt{ \|I\|_2})$.
\end{proof}

After arguing that $g$ is additive in $\ball(0,1/r)$, it will follow that $g$ is additive elsewhere because $g$ is defined by extrapolating the value of $g$ within this ball. Therefore, we will focus on proving the additivity of $g$ within $\ball(0,1/r)$.
\begin{lem}\label{lem:g-additive-in-ball}
  Suppose that~\eqref{eq:lem3-1} --~\eqref{eq:lem3-3} of Lemma~\ref{lem:f-additive-whp} hold. For every $p,q \in \mathbb{R}^n$ with $\|p\|_2, \|q\|_2, \|p+q\|_2 \leq 1/r$ it holds that $g(p+q) =g(p)+g(q)$.
\end{lem}

The proof of this lemma will crucially rely on the following two lemmas, which say that the conclusions of Lemma~\ref{lem:f-additive-whp} hold with high probability even when one of the points are fixed to a point $\ball(0,1/r)$. A consequence of this is that $g$ is well-defined.
\begin{lem}\label{lem:g-is-well-defined}
  Suppose that~\eqref{eq:lem3-1} --~\eqref{eq:lem3-3} of Lemma~\ref{lem:f-additive-whp} hold, then $g$ is well-defined, and for every $p \in \mathbb{R}^n$ with $\|p\|_2 \leq 1/r$,
  \[ \Pr_{x \sim \mathcal{N}(0,I)}[g(p) = f(p-x) +f(x)] \geq \frac{9}{10}. \]
\end{lem}
\begin{proof}
  Fix a point $p \in \mathbb{R}^n$ with $\|p\|_2 \leq 1/r$. We will bound the following probability.
  \[ A:= \Pr_{x,y \sim \mathcal{N}(0,I)} [f(p-x) + f(x) = f(p -y) +f(y)].\]
  Observe that
  \begin{align*}
    A=& \Pr_{x,y \sim \mathcal{N}(0,I)}[f(x) - f(y) \neq f(p - y) - f(p-x) ]\\
    \leq & \Pr_{x,y \sim \mathcal{N}(0,I)}[f(x) - f(y) \neq f(x-y) ] + \Pr_{x,y \sim \mathcal{N}(0,I)}[f(x-y) \neq f(p-y) - f(p-x)] \\
    < & \frac{1}{100} + \Pr_{x,y \sim \mathcal{N}(0,I)}[f(x-y) \neq f(p-y) - f(p-x)] \tag{By Lemma~\ref{lem:f-additive-whp}}
  \end{align*}
  It remains to bound the second term. Intuitively, because $x-p, y-p \sim \mathcal{N}(-p, I)$ and  $p \approx 0$, the random variables $p-x$ and $p-y$ should be distributed similarly to $x$ and $y$. Indeed,
  \begin{align*}
    &\Pr_{x,y \sim \mathcal{N}(0, I)}[f(x-y) \neq f(p-y) - f(p-x)] \\ =& \Pr_{x,y \sim \mathcal{N}(0, I)}[f(x-p + p-y) \neq f(p-y) - f(p-x)] \\
    =&\Pr_{x,y \sim \mathcal{N}(-p, I)}[f(x-y) \neq f(-y) - f(-x)] \\
    \leq & \Pr_{x,y \sim \mathcal{N}(0, I)}[f(x-y) \neq f(-y) - f(-x)] +2 \dtv \Big(\mathcal{N}(0, I), \mathcal{N}(-p, I) \Big) \\
    \leq & \Pr_{x,y \sim \mathcal{N}(0, I)}[f(x-y) \neq f(x) - f(y)] + \frac{2}{100} + 2 \Pr_{x \sim \mathcal{N}(0, I)}[ f(-x) \neq f(x)] \tag{Claim~\ref{clm:tvd-bound-on-ball}} \\
    \leq & \frac{3}{100} +  \frac{2}{100} = \frac{5}{100}. \tag{By~\eqref{eq:lem3-1} and~\eqref{eq:lem3-2} in Lemma~\ref{lem:f-additive-whp}}
  \end{align*}
  Plugging this into our previous bound on $A$, we can conclude that
  \[ A \geq 1- \left(\frac{1}{100} + \frac{5}{100} \right) = 1 - \frac{6}{100} > \frac{9}{10}.  \]

  Next, we bound $A$ above in terms of the probability that $g(p) \neq f(p-x) + f(x)$. Define $P_p \colon \mathbb{R}^n \rightarrow \mathbb{R}^+$ to be the bounded Lebesgue-measurable function such that $\int_B P_p(x) dx$ is the probability that $f(p -x) + f(x)$ takes value in the (measurable) set $B$. By H{\"o}lder's inequality with $p=1,q= \infty$ we have
  \[ A = \int_{\mathbb{R}} P_p^2(x) dx \leq \|P_p\|_\infty \int_{\mathbb{R}} P_p(x) dx = \|P_p\|_\infty, \]
  where the last equality follows because $P_p$ is a density and $\int_{\mathbb{R}} P_p(x) dx = 1$ holds.
  Therefore,
  \[ \frac{9}{10} \leq A  \leq  \| P_p \|_{\infty}. \]
  Because $\argmax_x P_p(x) \geq 9/10 > 1/2$, we have $g(p) = \argmax_x P_p(x)$ and
  hence $\Pr_{x \sim \mathcal{N}(0, I)} [g(p) = f(p -x) +f(x)] \geq 9/10$.
\end{proof}

The following lemma is essentially condition~\eqref{eq:lem3-3} of Lemma~\ref{lem:f-additive-whp} with two fixed points.

\begin{lem}\label{lem:whp-g-additive-two-points}
  Suppose that~\eqref{eq:lem3-1} --~\eqref{eq:lem3-3} of Lemma~\ref{lem:f-additive-whp} hold then, for every $p,q \in \mathbb{R}^n$ with $\|p\|_2, \|q\|_2, \|p+q\| \leq 1/r$,
  \[ \Pr_{x,y,z \sim \mathcal{N}(0, I)} \Big[ g(p + q) \neq f \Big(p - \frac{x-z}{2} \Big) + f \Big( q - \frac{z-y}{2} \Big) + f \Big( \frac{x-y}{2} \Big) \Big] \leq  \frac{2}{10}. \]
\end{lem}

\begin{proof}
  Fix a pair of points $p,q \in \mathbb{R}^n$ with $\|p\|_2, \|q\|_2 \leq 1/r$. We can bound the probability
  \begin{align*}
    &\Pr_{x,y,z \sim \mathcal{N}(0, I)} \Big[ g(p + q) \neq f \Big(p - \frac{x-z}{2} \Big) + f \Big( q - \frac{z-y}{2} \Big) + f \Big( \frac{x-y}{2} \Big) \Big] \\
    \leq &\Pr_{x,y,z \sim \mathcal{N}(0, I)} \Big[ g(p+q) \neq f \Big( p+q - \frac{x-y}{2} \Big) + f \Big( \frac{x-y}{2} \Big) \Big] \\
    &+ \Pr_{x,y,z \sim \mathcal{N}(0, I)} \Big[ f \Big( p + q - \frac{x-y}{2} \Big) \neq f \Big( p - \frac{x-z}{2} \Big) + f \Big( q -\frac{z-y}{2} \Big) \Big]
  \end{align*}
  To bound the first term, observe that if $x,y \sim \mathcal{N}(0, I)$, then the random variable $(x-y)/2$ is also distributed according to $\mathcal{N}(0, I)$. Furthermore, because $\|p+q \|_2 \leq 1/r$,  we can apply Lemma~\ref{lem:g-is-well-defined} and conclude that
  \begin{align*}
    \Pr_{x,y,z \sim \mathcal{N}(0, I)} \Big[ g(p + q) \neq f \Big(p + q -  \frac{x-y}{2} \Big) + f \Big(\frac{x-y}{2} \Big) \Big] \leq \frac{1}{10}.
  \end{align*}
  To bound the second term, observe that
  \begin{align*}
    &\Pr_{x,y,z \sim \mathcal{N}(0, I)} \Big[ f \Big( p + q - \frac{x-y}{2} \Big) \neq \Big(p- \frac{x-z}{2} \Big) + \Big(q- \frac{z-y}{2} \Big) \Big] \\
    =&\Pr_{x,y,z \sim \mathcal{N}(0, I)} \left[ f \left(\frac{(2q +y)-(x-2p)}{2} \right) \neq \left( \frac{(2q +y)-z}{2} \right) + \left( \frac{z-(x-2p)}{2} \right) \right] \\
    =&\Pr_{\substack{x \sim \mathcal{N}(-2p,I) \\ y \sim \mathcal{N}(2q,I) \\ z \sim \mathcal{N}(0,1)}}\Big[ f \Big(\frac{y-x}{2} \Big) \neq \Big( \frac{y-z}{2} \Big) + \Big( \frac{z-x}{2} \Big) \Big] \\
    \leq  &\Pr_{x,y,z \sim \mathcal{N}(0, I)}\Big[ f \Big(\frac{x-y}{2} \Big) \neq \Big( \frac{x-z}{2} \Big) + \Big( \frac{z-y}{2} \Big) \Big]  + \dtv \Big( \mathcal{N}(0, I), \mathcal{N}(-2p, I )\Big) \\&+  \dtv \Big( \mathcal{N}(0,I), \mathcal{N}(2q, I )\Big) \\
    \leq & \frac{1}{100} + \frac{2}{100} + \frac{2}{100} = \frac{5}{100}. \tag{By Lemma~\ref{lem:f-additive-whp} and Claim~\ref{clm:tvd-bound-on-ball}}
  \end{align*}
  Combining both of these bounds, we have
  $\Pr_{x,y,z \sim \mathcal{D}} [ g(p + q) \neq f (p - \frac{x-z}{2} ) + f ( q - \frac{z-y}{2} ) + f (\frac{x-y}{2} ) ]
  \leq 1/10+5/100 \leq 2/10$.
\end{proof}

The additivity of $g$ within $\ball(0,1/r)$ is an immediate consequence of these two lemmas.

\begin{proof}[Proof of Lemma~\ref{lem:g-additive-in-ball}]
  Let $p,q \in \mathbb{R}^n$ be any pair of points satisfying $\|p\|_2, \|q\|_2, \|p+q\|_2 \leq 1/r$. Our aim is to show that $g(p + q) = g(p) + g(q)$. By a union bound over Lemmas~\ref{lem:g-is-well-defined} and~\ref{lem:whp-g-additive-two-points}, the probability that $x,y,z \sim \mathcal{N}(0,I)$ simultaneously satisfy
  \begin{enumerate}
    \item $g(p + q) = f(p - \frac{x-z}{2}) + f( q - \frac{z-y}{2} ) + f(\frac{x-y}{2})$,
    \item $g(p) = f(p - \frac{x-z}{2} ) + f(\frac{x-z}{2})$,
    \item $g(q) = f(q - \frac{z-y}{2}) + f( \frac{z-y}{2})$,
    \item $f(\frac{x-y}{2}) = f(\frac{x-z}{2}) - f ( \frac{z-y}{2})$
  \end{enumerate}
  is at least $1- (2/10+ 2\cdot 1/10  + 1/10)  > 0$. Here we are using the fact that $((x-y)/2)$ is distributed as $\mathcal{N}(0,I)$. Fixing such a triple $(x,y,z)$, we conclude that
  \begin{align*}
    g ( p + q) &= f \Big(p - \frac{x-z}{2} \Big) + f \Big( q - \frac{z-y}{2} \Big) + f \Big (\frac{x-y}{2} \Big) \\
    &=g(p) + g(q) + f \Big (\frac{x-y}{2} \Big) - f \Big(\frac{x-z}{2} \Big) -  f \Big( \frac{z-y}{2} \Big) \\
    &=g(p) + g(q).
  \end{align*}
  Therefore $g$ is additive within $\ball(0,1/r)$.
\end{proof}

Finally, we argue that $g$ is additive everywhere. Intuitively this should be true because the values of $g$ on points outside of $\ball (0,1/r)$ are defined by extrapolating the values of $g$ on points within $\ball (0,1/r)$, where we know $g$ is additive. For the proof, it will be useful to record the following fact.
 \begin{fact}\label{fact:g-homogeneous-over-ball}
 	 Provided that~\eqref{eq:lem3-1} --~\eqref{eq:lem3-3} of Lemma~\ref{lem:f-additive-whp} hold then, for  every $p \in \mathbb{R}^n$ with $\| p \|_2 \leq 1/r$ and $c \in \mathbb{Z}^{> 0}$, we have  $g(p) = c g(p/c)$.
 \end{fact}
 \begin{proof} Observe that
  $g(p) = g((c/c)p) = g \left( \sum_{i=1}^c p/c \right) = \sum_{i=1}^c g(p/c) =  c \cdot g(p/c)$, where the third equality follows by Lemma~\ref{lem:g-additive-in-ball}, noting that $\|kp/c\|_2 \leq 1/r$ for every $k \in [c-1]$
 \end{proof}

\begin{proof}
  [Proof of Lemma~\ref{lem:g-is-additive}]
  Fix a pair of points $p,q \in \mathbb{R}^n$, we will argue that $g(p+q) = g(p) + g(q)$. Recall that $g(p):= k_p g(p/k_p)$, $g(q):= k_q g(q/k_q)$, and $g(p+q):= k_{p+q} g((p+q)/k_{p+q})$. Then,
  \[ g(p)+g(q) = k_p \cdot g\left( \frac{p}{k_p} \right) + k_q \cdot g\left( \frac{p}{k_q} \right) = k_p k_q k_{p+q} \cdot g\left( \frac{p}{k_p k_q k_{p+q}} \right) + k_p k_q k_{p+q} \cdot g\left( \frac{p}{k_p k_q k_{p+q}} \right),  \]
  where the second equality follows by Fact~\ref{fact:g-homogeneous-over-ball}, noting that $k_p, k_q, k_{p+q} \in \mathbb{Z}^{> 0}$ and so $p/k_p, q/k_q \in \ball(0,1/r)$. Furthermore, because $p/(k_p k_q k_{p+q}), q/(k_p k_q k_{p+q}),(p+q)/ (k_p k_q k_{p+q}) \in \ball(0,1/r)$, we can apply Lemma~\ref{lem:g-additive-in-ball} to obtain
  \begin{align*} k_p k_q k_{p+q} \left( g\left( \frac{p}{k_p k_q k_{p+q}} \right) + g\left( \frac{p}{k_p k_q k_{p+q}} \right) \right) &= k_p k_q k_{p+q} \cdot g\left( \frac{p+q}{k_p k_q k_{p+q}} \right) \\
  &= k_{p+q} \cdot g\left( \frac{p+q}{k_{p+q}} \right) \\
  &=g(p+q),
  \end{align*}
  where the second equality follows by Fact~\ref{fact:g-homogeneous-over-ball}, noting that $k_p k_q \in \mathbb{Z}^{>0}$ and $(p+q)/k_{p+q} \in \ball(0,1/r)$.

  Finally, by Lemma~\ref{lem:g-is-well-defined}, $g$ is well-defined within $\ball(0,1/r)$. Because $g$ is defined by extrapolating from its value within this ball, it is well-defined everywhere.
\end{proof}

\begin{remark}
	This tester (and the same proof) will in fact work over any Gaussian $\mathcal{N}(0,\Sigma)$ for arbitrary covariance matrix $\Sigma \in \mathbb{R}^{n \times n}$ by setting the value of $r$ to be $50\sqrt{ \|\Sigma^{-1} \|_2}$.
\end{remark}

\subsection{Distribution-Free Tester}\label{sec:distribution-free-additivity}

In this section, we prove Theorem~\ref{thm:intro-additivity} by adapting our tester for additivity over the standard Gaussian (Algorithm~\ref{alg:zero-mean-additivity}) to a distribution-free tester.

Assuming that we are able to draw samples from the standard Gaussian (or in fact any Gaussian), the modification to Algorithm~\ref{alg:zero-mean-additivity} is straight forward. Indeed, we will only have to modify Algorithm~\ref{alg:zero-mean-additivity}, the two subroutines will remain the same. Let $\mathcal{D}$ be our unknown distribution by which we will measure the distance of $f$ to an additive function. The high-level idea is to first run the \Call{TestAdditivity}{} subroutine over the standard Gaussian. If this passes, then we know that with high probability $g$ is additive. We can obtain query access to $g(p)$ (with high probability) as before by sampling points $x \sim \mathcal{N}(0,I)$ and checking that the values of $k_p(f(p/k_p -x) + f(x))$ agree for all of the $x$ that we sample. To test whether $f$ and $g$ are $\varepsilon$-far according to $\mathcal{D}$ it suffices to sample points $p \sim \mathcal{D}$ and check whether $f(p)$ and $g(p)$ agree.

Our algorithm is given in Algorithm~\ref{alg:distribution-free-additivity}.
We stress that both subroutines \Call{TestAdditivity}{} and \Call{Query-$g$}{$p_i$} are being performed over $\mathcal{N}(0,I)$, i.e., they do not use $\mathcal{D}$.

\begin{algorithm}[t!]
  \caption{Distribution-Free Additivity Tester}\label{alg:distribution-free-additivity}
  \Given{query access to $f: \mathbb{R}^n \rightarrow \mathbb{R}$, sampling access to an unknown distribution $\mathcal{D}$, and sampling access to $\mathcal{N}(0, I)$;}
  \textbf{Reject} if \Call{TestAdditivity}{$f$} returns \textbf{Reject}\;
  \For{$N_{\ref{alg:distribution-free-additivity}} := O(1/\varepsilon)$ times}{
    Sample $p \sim \mathcal{D}$\;
    \textbf{Reject} if $f(p) \neq$ \Call{Query-$g$}{$p,f$} or if \Call{Query-$g$}{$p,f$} returns \textbf{Reject}.
  }
  \textbf{Accept}.
\end{algorithm}


\begin{proof}[Proof of Theorem~\ref{thm:intro-additivity}]
  The proof is nearly identical to the proof of Theorem~\ref{thm:central-gausssian}. Again, observe that if $f$ is an additive function then Algorithm~\ref{alg:distribution-free-additivity} always accepts.

  It remain to show that if $f$ is $\varepsilon$-far from additive functions, then Algorithm~\ref{alg:distribution-free-additivity} rejects with probability at least $2/3$.
  If \Call{TestAdditivity}{$f$} accepts with probability at most $1/10$, we can reject $f$ with probability at least $1-1/10 > 2/3$.
  Hence, we assume that \Call{TestAdditivity}{$f$} accepts with probability at least $1/10$.
  By Lemma~\ref{lem:g-is-additive}, the function $g$ is additive and hence $f$ is $\varepsilon$-far from $g$.
  Note that the probability that \Call{Query-$g$}{$p,f$} fails to correctly recover $g(p)$ is at most $\varepsilon/2$ by the same argument as before.
  It remains to bound the probability that Step~3 fails to reject, which is
  \begin{align*}
     &{\left( \Pr_{p \sim \mathcal{N}(0,I)} \left[ f(p) = g(p) \vee \text{\Call{Query-$g$}{$p$} fails to correctly recover $g(p)$} \right]  \right)}^{N_{\ref{alg:distribution-free-additivity}}}
    < {\left(1- \frac{\varepsilon}{2}\right)}^{N_{\ref{alg:distribution-free-additivity}}} < \frac{1}{10},
  \end{align*}
  by choosing the hidden constant in $N_{\ref{alg:distribution-free-additivity}}$ to be large enough, by the same argument as before. Therefore, Algorithm~\ref{alg:distribution-free-additivity} rejects with probability at least $1-1/10> 2/3$.
\end{proof}

\section{Testing Linearity of Continuous Functions}\label{sec:linearity}

In this section, we prove Theorem~\ref{thm:intro-linearity} by adapting the tester from the previous section (Algorithm~\ref{alg:distribution-free-additivity}) to test whether $f$ is linear, given that $f$ is a continuous function.


We would like to argue that if  $f$ is continuous and Algorithm~\ref{alg:distribution-free-additivity} passes then $g$ is in fact a linear function with high probability. However, in order to exploit continuity, we need $f$ to satisfy $f(-x) = -f(x)$ for every $x \in \mathbb{R}^n$. First, we will show how to argue that $g$ is linear assuming that $f(-x) = -f(x)$. After that, we will handle the case when this property does not hold.

 \begin{lem}\label{lem:continuous-f-linear-g}
  If $f\colon\mathbb{R}^n \rightarrow \mathbb{R}$ is a continuous function satisfying $f(-x) = -f(x)$ and the assumptions of Lemma~\ref{lem:g-is-additive} hold, then
  the function $g$ is linear.
\end{lem}
The proof will rely on the following claim which was originally proved by Darboux in 1875.
\begin{claim}\label{clm:continuous-implies-linear}
    Any additive function $f\colon\mathbb{R}^n \rightarrow \mathbb{R}$ which is continuous at a point $x_0 \in \mathbb{R}^n$ is a linear function.
  \end{claim}
  \begin{proof}
  		First, it is well-known that any additive function which is continuous at a point is continuous everywhere (see e.g.,~\cite{bartle2011introduction}). Next, we argue that the continuity of $f$ implies that $f(rx) = rf(x)$ for every $r \in \mathbb{R}$ and $x \in \mathbb{R}^n$. Because $f$ is additive, this homogeneity holds for every $r \in \mathbb{Q}$, so it suffices to assume that $r$ is irrational.

      Fix $x \in \mathbb{R}^n$ and irrational $r$.
  		Then for any $\zeta > 0$, we can always find $\tilde{r} \in \mathbb{Q}$ such that $|\tilde{r} - r| < \zeta$ and $\|\tilde{r} x - rx \|_2 < \zeta$. Now, by the continuity of $f$, for any $\xi > 0$ there exists $\zeta > 0$ such that whenever  $\| \tilde{r}x -rx\|_2 < \zeta$, we have $|f(\tilde{r} x ) -f(rx)| < \xi$.
  		Now, take a sequence ${\{ \xi_i\}}_i$ with $\xi_i \rightarrow 0$ and consider the corresponding sequence ${\{\zeta_i\}}_i$ with $\zeta_i \rightarrow 0$. Let ${\{ \tilde{r}_i \}}_i$ with $r_i \in \mathbb{Q}$ be the sequence of approximations such that $|\tilde{r}_i -r| \leq \zeta_i$ and $\|\tilde{r}_i x -rx\|_2 \leq \zeta_i$. Then,
  		\[ |f(rx) - rf(x)| \leq |f(rx) - f(\tilde{r}_i x)| + |f(\tilde{r}_i x) - rf(x)| \leq \xi_i + |\tilde{r}_i f(x) - rf(x)| \leq \xi_i + \zeta_i |f(x)|. \]
  		Because $\zeta_i, \xi_i \rightarrow 0$, $|f(rx) - rf(x)| \rightarrow 0$ and so $f(rx) = rf(x)$.
  \end{proof}
  With this claim in hand, we are ready to prove Lemma~\ref{lem:continuous-f-linear-g}.

\begin{proof}[Proof of Lemma~\ref{lem:continuous-f-linear-g}]
	Let $f$ be a continuous function satisfying $f(-x) = -f(x)$. By Lemma~\ref{lem:g-is-additive}, the function $g$ is additive. Conditioned on this event, we will show that the continuity of $f$ implies that $g$ is linear as well. To do so, we will argue that $g$ is continuous at the origin and then appeal to Claim~\ref{clm:continuous-implies-linear} to conclude that $g$ is linear.

	Let $B$ be a ball of mass $1/2$ (with respect to $\mathcal{N}(0,I)$) centred at the origin.
	Let ${\{p_i\}}_i$ be any sequence of points with $p_i \in B$, $\|p_i \|_2 \leq 1/r$ and $p_i \rightarrow 0$. Now, let ${\{x_i\}}_i$ be a sequence of points such that $g(p_i) = f(p_i - x_i) +f(x_i)$ and $x_i \in B$. Such a sequence exists because, by Lemma~\ref{lem:g-is-additive} $\Pr_{x \sim \mathcal{N}(0,I)}[g(x) = f(p_i-x) + f(x)] \geq 1/2$ and so for every $p_i$ there must exist such an $x_i$ in $B$.

	Let $S$ be the ball centred at the origin with twice the radius of $B$.
  As $S$ is compact and $f$ is continuous, $f$ is uniformly continuous on $S$.
	Thus for every $\xi > 0$, there exists $\zeta >0$ such that $|f(p_i - x_i) - f(-x_i)| = |f(p_i - x_i) + f(x_i) | < \xi$ whenever $\|(p_i - x_i) + x_i \|_2 < \zeta$.
  Now, take a sequence ${\{ \xi_i \}}_i$ with $\xi_i \rightarrow 0$ and consider the corresponding sequence ${\{ \zeta_i \}}_i$.
  As $p_i \rightarrow 0$, for every $i$, there exists $j$ such that $\|(p_j -x_j) + x_j \|_2 < \zeta_i$ which in particular implies that $|g(p_j) | = |f(p_j - x_j) + f(x_j) | < \xi_i$. Thus, $g(p_i) \rightarrow 0$, and $g$ is continuous at the origin. By Claim~\ref{clm:continuous-implies-linear}, we can conclude that $g$ is a linear function.
\end{proof}

 Now we consider the case when $f(-x) \neq -f(x)$ for some $x$.
 Luckily, in this case we can \emph{force} $f$ to satisfy $f(-x) = -f(x)$. To do so, we test whether $f$ is $\varepsilon /2$-far from satisfying this property. If it is, then we reject $f$, otherwise, we can replace $f$ with a function $f'$ guaranteed to satisfy this property, by defining
 \begin{align*}
	 f'(x) := \frac{f(x)- f(-x)}{2}.
 \end{align*}
We then continue to work over $f'$ rather than $f$.
Our modified algorithm is given in Algorithm~\ref{alg:distribution-free-linearity}, which uses Algorithm~\ref{alg:forcenegativity-subroutine} as a subroutine.

\begin{algorithm}[t!]
  \caption{Distribution-Free Linearity Tester}\label{alg:distribution-free-linearity}
  \Given{query access to a continuous $f\colon \mathbb{R}^n \rightarrow \mathbb{R}$, sampling access to an unknown distribution $\mathcal{D}$, and sampling access to $\mathcal{N}(0, I)$;}
  \textbf{Reject} if \Call{ForceNegativity}{$f, \mathcal{D}$} returns \textbf{Reject}\;
  Let $f'$ be the returned function\;
  \textbf{Reject} if \Call{TestAdditivity}{$f'$} returns \textbf{Reject}\;
  \For{$N_{\ref{alg:distribution-free-linearity}} := O(1/\epsilon)$ times}{
    Sample $p \sim \mathcal{D}$\;
    \textbf{Reject} if $f'(p) \neq$ \Call{Query-$g$}{$f', p$} or if \Call{Query-$g$}{$f',p$} returns \textbf{Reject}.
  }
  \textbf{Accept}.
\end{algorithm}

\begin{algorithm}[!t]
  \caption{Force Negativity Subroutine}\label{alg:forcenegativity-subroutine}
\Procedure{\emph{\Call{ForceNegativity}{$f, \mathcal{D}$}}}{
    \Given{query-Access to $f\colon \mathbb{R}^n \rightarrow \mathbb{R}$ and sampling access to an unknown distribution $\mathcal{D}$;}
  	\For{$N_{\ref{alg:forcenegativity-subroutine}} := O(1/\epsilon)$ times}{
      Sample $x \sim \mathcal{D}$\;
    	\textbf{Reject} if $f(-x) \neq -f(x)$\;
  	}
  	\textbf{Return} a function $f'\colon \mathbb{R}^n \to \mathbb{R}$ where $f'(x):=\frac{f(x)- f(-x)}{2}$\;
  }
\end{algorithm}


\begin{claim}\label{clm:f-forced-close}
 	If \Call{ForceNegativity}{$f,\mathcal{D}$} accepts with probability at least $1/10$, then $\Pr_{x\sim \mathcal{D}}[f(x) =  f'(x)] \geq 1-\varepsilon$.
\end{claim}
\begin{proof}
  Suppose for contradiction that $\Pr_{x\sim \mathcal{D}}[f(x) =  f'(x)] \leq 1-\varepsilon$.
 	Observe that for a point $x \in \mathbb{R}$, $f'(x) \neq f(x)$ iff $f(-x) \neq -f(x)$. Therefore, by choosing the hidden constant in $N_{\ref{alg:forcenegativity-subroutine}}$ to be large enough, the probability that all the sampled points $x$ satisfy $f(x) = -f(x)$ is at most
 	\[ {\Big( \Pr_{x \sim \mathcal{D}}[f(-x) = f(x)] \Big)}^{N_{\ref{alg:forcenegativity-subroutine}}} < {(1- \varepsilon)}^{N_{\ref{alg:forcenegativity-subroutine}}} \leq \frac{1}{10}, \]
  which is a contradiction.
 \end{proof}
 Therefore if \Call{ForceNegativity}{$f,\mathcal{D}$} accepts with probability at least $1/10$, $f$ and $f'$ are $\varepsilon/2$-close. Furthermore, because $f$ is continuous and $f'$ is the sum of continuous functions, $f'$ is continuous as well, and so we can proceed with $f'$ in place of $f$.

 \begin{proof}[Proof of Theorem~\ref{thm:intro-linearity}]
 	First observe that if $f$ is linear then $f = f'$ and Algorithm~\ref{alg:distribution-free-linearity} always accepts.

  Now, we show that if $f$ is $\varepsilon$-far from linear functions, then Algorithm~\ref{alg:distribution-free-linearity} rejects with probability at least $2/3$.
  If either the \Call{TestAdditivity}{} subroutine or the \Call{ForceNegativity}{} subroutine passes with probability at most $1/10$, we can reject $f$ with probability at least $1-1/10>2/3$.
  Hence, we assume both the subroutines pass with probability at least $1/10$.
  Then by Lemma~\ref{clm:f-forced-close}, $f$ is $\varepsilon/2$-close to $f'$, which means that $f'$ is $\varepsilon/2$-far from linear.
  Also by Lemma~\ref{lem:continuous-f-linear-g}, because $f'$ is continuous and satisfies $f'(-x) = -f'(x)$, the function $g$ is linear, and so $f'$ is $\varepsilon/2$-far from $g$.
  Therefore, Algorithm~\ref{alg:distribution-free-linearity} rejects $f$ with probability at least $1-1/10 > 2/3$.
 \end{proof}



\section{Lower Bounds on Testing Linearity in the Sampling Model}\label{sec:lower-bound}

In this section, we prove Theorem~\ref{thm:intro-lower-bound}, that is, we show without query access, any tester requires a linear number of samples in order to test linearity and additivity over the standard Gaussian.
We note that we can obtain the same lower bound for testing additivity just by replacing linearity with additivity in the proof.


By Yao's minimax principle it suffices to construct two distributions, $\mathcal{D}_\mathrm{yes}$ over linear functions and $\mathcal{D}_\mathrm{no}$ over functions which are (with probability $1$) $1/3$-far from linear such that any deterministic $n$-sample algorithm cannot distinguish between them with probability at least $2/3$. Let $\delta \in \mathbb{R}^{\geq 0}$ be some  parameter to be set later; we will think of $\delta$ as tiny. Instances from these two distributions are generated as follows:
\begin{enumerate}
	\item[$\mathcal{D}_\mathrm{yes}$]: Sample $w \sim \mathcal{N}(0,I)$ and return $f(x) :=\langle w,x \rangle$.
	\item[$\mathcal{D}_\mathrm{no}$]: Sample $w \sim \mathcal{N}(0,I)$ and for every $x \in \mathbb{R}^n$ sample $\varepsilon_x \sim {\mathcal{N}(0, \delta)}$. Return $f(x) := \langle w, x \rangle +\varepsilon_x$.
\end{enumerate}
The functions in the support of $\mathcal{D}_\mathrm{yes}$ are linear by definition. It remains to show that the instances in the support of $\mathcal{D}_\mathrm{no}$ are far from linear.
\begin{lem}\label{lem:no-far-from-linear}
	With probability $1$ any $f \sim \mathcal{D}_\mathrm{no}$ is $1/3$-far from linear.
\end{lem}
The proof of this lemma will hinge on the following claim.

\begin{claim}\label{clm:rows-contain-one-bad}
  Let $f \sim \mathcal{D}_\mathrm{no}$, for $x,y,z \sim \mathcal{N}(0,1)$,
  $ \Pr[f(\frac{x-y}{2}) \neq f(\frac{x-z}{2}) + f( \frac{z-y}{2}) ] = 1$.
\end{claim}
\begin{proof}
  Observe that
  $\Pr \left[f \left(\frac{x-y}{2} \right) = f \left(\frac{x-z}{2} \right) + f \left( \frac{z-y}{2} \right) \right] = \Pr \left[ \varepsilon_{(x-y)/2} = \varepsilon_{(x-z)/2}+  \varepsilon_{(z-y)/2}  \right]$, where the probability is over $\varepsilon_{(x-y)/2}, \varepsilon_{(x-z)/2}, \varepsilon_{(z-y)/2} \sim \mathcal{N}(0,\delta)$. Define the random variable $z:=\varepsilon_{(x-y)/2} - \varepsilon_{(x-z)/2} -\varepsilon_{(z-y)/2}$, and note that $z$ is distributed according to $\mathcal{N}(0,3 \delta)$. Then \[\Pr_{z \sim \mathcal{N}(0,3)} \left[ \varepsilon_{(x-y)/2} = \varepsilon_{(x-z)/2}+  \varepsilon_{(z-y)/2}  \right] = \Pr_{z \sim \mathcal{N}(0,3)}[z = 0].\]
  By standard arguments, we have $\Pr_{z \sim \mathcal{N}(0,3 \delta)}[z = 0]  =0$.
\end{proof}

\begin{proof}[Proof of Lemma~\ref{lem:no-far-from-linear}]
	Let $f^*$ be the closest linear function to $f$. For a point $x \in \mathbb{R}^n$, say that $f(x)$ is \emph{bad} if $f(x) \neq f(x^*)$.
Construct the following matrix: the rows are labelled by every triple $(\frac{x-y}{2}, \frac{x-z}{2}, \frac{z-y}{2})$ and there are three columns. The entries at row $(\frac{x-y}{2}, \frac{x-z}{2}, \frac{z-y}{2})$ are $f(\frac{x-y}{2})$, $f(\frac{x-z}{2})$, and $f(\frac{z-y}{2})$. Note that because $x,y,z \sim \mathcal{N}(0,1)$, the points $\frac{x-y}{2}, \frac{x-z}{2}, \frac{z-y}{2}$ are distributed according to $\mathcal{N}(0,1)$.

Henceforth, we will measure mass in terms of probability mass over $\mathcal{N}(0,1)$.
By Claim~\ref{clm:rows-contain-one-bad}, the probability that each row contains a bad entry is $1$. Therefore, there must be some column for which the probability mass of the bad entries is at least $1/3$.
This implies that a mass of at least $1/3$ of $f$ must be changed to obtain $f^*$. Because $f^*$ is the closest linear function to $f$, this implies that $f$ is $1/3$-far from linear.
\end{proof}

Having defined our distributions over linear and far-from-linear functions, it remains to argue that no algorithm receiving $n$ samples can distinguish between them with high probability.

\begin{proof}[Proof of Theorem~\ref{thm:intro-lower-bound}]
	Let $\cal D$ be the distribution that with probability $1/2$ draws $f \sim \mathcal{D}_\mathrm{yes}$ and otherwise draws $f \sim \mathcal{D}_\mathrm{no}$. Let $A$ be any deterministic algorithm which receives $n$ samples $x_1,\ldots, x_n \sim \mathcal{N}(0,I)$.
	By Yao's minimax principle, it suffices to show that $A$ cannot correctly distinguish which distribution of the distributions $\mathcal{D}_\mathrm{yes}$ or $\mathcal{D}_\mathrm{no}$ a given sample $f \sim \mathcal{D}$ comes from with probability at least $2/3$. That is, we would like to show that
	\begin{align} \Big| \Pr_{\substack{f \sim \mathcal{D}_\mathrm{yes} \\ x_1,\ldots,x_n \sim \mathcal{N}(0,I)}}[A(f(x_1),\ldots,f(x_n)) = \mathrm{YES}] - \Pr_{\substack{f \sim \mathcal{D}_\mathrm{no} \\ x_1,\ldots,x_n \sim \mathcal{N}(0,I)}}[A(f(x_1),\ldots,f(x_n)) =\mathrm{YES}] \Big| \label{eq:to-bound}
	\end{align}
	is $o(1)$. Suppose for contradiction that an algorithm $A$ exists that with probability at least $2/3$ distinguishes these distributions.

	Observe that the~\eqref{eq:to-bound} can be bounded from above by the total variation distance between the distributions $(f^y(x_1),\ldots,f^y(x_n))$ for $f^y \sim \mathcal{D}_\mathrm{yes}$, and $(f^n(x_1),\ldots,f^n(x_n))$ for $f^n \sim \mathcal{D}_\mathrm{no}$, for $x_1, \ldots, x_n \sim \mathcal{N}(0,I)$, as applying the algorithm $A$ can only make the total variation distance smaller. By the definition of $\mathcal{D}_\mathrm{yes}$ and $\mathcal{D}_\mathrm{no}$, this means bounding the total variation distance between $(w_y^\top x_1, \ldots, w_y^\top x_n)$ and $(w_n^\top x_1 + \varepsilon_{x_1}, \ldots, w_n^\top x_n + \varepsilon_{x_n})$, where $w_y \sim \mathcal{D}_\mathrm{yes}$ and $w_n \sim \mathcal{D}_\mathrm{no}$

	Now, let $X \in \mathbb{R}^n$ be the matrix whose rows are $x_1, \ldots, x_n$. Because $w_y,w_n \sim \mathcal{N}(0,I)$ and $\varepsilon_{x_i} \sim \mathcal{N}(0,\delta)$, it follows that
	\begin{align*}
  (w^\top x_1, \ldots, w^\top x_n) &\sim \mathcal{N}(0, XX^\top), \\
  (w_n^\top x_1, \ldots, w_n^\top x_n) + (\varepsilon_{x_1}, \ldots, \varepsilon_{x_n}) & \sim \mathcal{N}(0, XX^\top + \delta I).
\end{align*}
	Therefore,
	\[ \eqref{eq:to-bound} \leq \dtv (\mathcal{N}(0, XX^\top),\mathcal{N}(0, XX^\top + \delta I) ).
\]
	To bound this distance we will appeal to Pinkser's inequality and Lemma~\ref{lem:TVD-bound-on-Gaussian}. Thus it will be useful to first record some facts about the covariance matrices of these distribution. First, we show that the rows of the matrix $X$ are linearly independent with high probability.
\begin{fact}\label{fact:linear-independence}
  $\Pr_{x_1, \ldots, x_n \sim \mathcal{N}(0,1)}[ \spn(x_1,\ldots, x_n) = \mathbb{R}^n ] = 1$. 
\end{fact}
It follows that the covariance matrices of these two distributions are positive definite with high probability.
\begin{claim}\label{clm:covariance-pd}
	With probability $1$ the matrices $XX^\top$ and $XX^\top+ \delta I$ are positive definite.
\end{claim}
\begin{proof}
  $XX^\top \succ 0$ is immediate from the fact that by Fact~\ref{fact:linear-independence}, the rows of $X$ are linearly independent with probability $1$. Let $\lambda_1, \ldots, \lambda_n$ be the eigenvalues of $XX^\top$.  To prove that $XX^\top + \delta I \succ 0$ note that adding $\delta I$ simply adds $\delta$ to each of the eigenvalues. Thus, the eigenvalues of $XX^\top + \delta I$ are all positive.
\end{proof}
With these facts in hand we turn to bounding the total variation distance between $\mathcal{N}(0, XX^\top)$ and $\mathcal{N}(0, XX^\top + \delta I)$. Denote by $\Sigma_{\mathrm{YES}}:= XX^\top $ and $\Sigma_{\mathrm{NO}} := XX^\top + \delta I$.
 By Pinkser's inequality (Theorem~\ref{thm:Pinskers}) and Lemma~\ref{lem:TVD-bound-on-Gaussian},
\[ \dtv \Big(\mathcal{N}(0, XX^\top), \mathcal{N}(0, XX^\top + \delta I) \Big)
\leq \sqrt{ \frac{1}{4} \left( \log \left( \frac{\det \Sigma_{\mathrm{YES}}}{ \det \Sigma_{\mathrm{NO}} } \right) + \trace \left( \Sigma_{\mathrm{YES}}^{-1} \Sigma_{\mathrm{NO}} \right) - n \right)  }. \]
We will bound each of these terms separately.
\paragraph{Bounding the Determinant.}
For simplicity of notation, we will bound the inverse of $\det (\Sigma_{\mathrm{YES}})/\det (\Sigma_\mathrm{NO})$ below. We have
\begin{align*} \frac{\det \Sigma_{\mathrm{NO}}}{ \det \Sigma_{\mathrm{YES}} } &= \frac{\det ( XX^\top + \delta I)}{\det (XX^\top)} \\
&= \det \left(XX^\top{\left(XX^\top\right)}^{-1} +\delta {\left(XX^\top\right)}^{-1} \right) \\ 
&=\det \left(I + \delta{\left(XX^\top\right)}^{-1} \right).
\end{align*}

\begin{claim}
  If $A$ is a diagonalizable matrix with eigenvalues $\lambda_1, \ldots, \lambda_n$ then $\det (A +  I) = \prod_{i=1}^n (\lambda_i + 1)$.
\end{claim}

Applying this claim, we have $\det(I + \delta{(XX^\top)}^{-1}) = (\delta \lambda_1^{-1} + 1) \ldots (\delta \lambda_n^{-1}+1)$, where $\lambda_1, \ldots, \lambda_n$ are the eigenvalues of $XX^\top$. By Claim~\ref{clm:covariance-pd} the matrix $XX^\top$ is positive definite and so $\lambda_i > 0$ for all $i$. Therefore, $(\delta \lambda_i^{-1} + 1) >1$ for all $i$, and we can conclude that $\det \Sigma_{\mathrm{NO}} / \det \Sigma_{\mathrm{YES}} > 1$. Thus we can upper bound $\det \Sigma_{\mathrm{YES}}/ \det \Sigma_{\mathrm{NO}}$ by $1$.


\paragraph{Bounding the Trace.} Next, we bound
\begin{align*} \trace \left(\Sigma_{\mathrm{YES}}^{-1} \Sigma_{\mathrm{NO}} \right) &= \trace \left( {\left(XX^\top\right)}^{-1}(XX^\top+\delta I) \right) \\
 &= \trace \left( I + \delta {\left(X^\top\right)}^{-1}X^{-1} \right) \\\
 &\leq \trace (I) + \delta \trace \left( {\left(X^\top\right)}^{-1}X^{-1} \right) \\
 &= n + \delta \sum_{i,j} {(X_{i,j}^{-1})}^2 \\
 &\leq n+ \delta n^2 \cdot  {\lambda_{\max}(X^{-1})}^2,
\end{align*}
where $\lambda_{\max}$ is the largest eigenvalue of $X^{-1}$.
Noting that the eigenvalues of $X^{-1}$ are the inverse of the eigenvalues of $X$, we have $\trace \left(\Sigma_{\mathrm{YES}}^{-1} \Sigma_{\mathrm{NO}} \right) \leq n + \delta n^2/ {\lambda_{\min}(X)}^2$. Setting $\delta := C {\lambda_{\min}(X)}^2/ n^2$ for some tiny $C > 0$ to be set later, we can conclude that $\trace \left(\Sigma_{\mathrm{YES}}^{-1} \Sigma_{\mathrm{NO}} \right) \leq n + C$.

\paragraph{Completing the proof.}

Putting our previous bounds together we conclude that
\[ \dtv \Big(\mathcal{N}(0, XX^\top), \mathcal{N}(0, XX^\top + \delta I) \Big) \leq \sqrt{ \frac{1}{4} \left(  \log(1) + n+C - n \right)} = \frac{1}{2}C^{1/2}.\]
By our previous argument we have
\[ (\ref{eq:to-bound}) \leq \dtv \Big(\mathcal{N}(0, XX^\top), \mathcal{N}(0, XX^\top + \delta I) \Big) \leq \frac{1}{2}C^{1/2}. \]
Setting $C < {(2/3)}^2$ contradicts our assumption of the existence of an algorithm $A$ which distinguishes a sample drawn from $\mathcal{D}_\mathrm{yes}$ from one drawn from $\mathcal{D}_\mathrm{no}$ with probability at least $2/3$, completing the proof.
\end{proof}

Finally, observe that the same proof goes through for testing additivity as well. Indeed, $\mathcal{D}_\mathrm{yes}$ is supported on additive functions, while $\mathcal{D}_\mathrm{no}$ is supported on functions which are far from additive with probability $1$. 

\begin{cor}
	Any sampler for additivity of functions $f\colon \mathbb{R}^n \rightarrow \mathbb{R}$ requires $\Omega(n)$ samples when $\mathcal{D} = \mathcal{N}(0,I)$.
\end{cor}

\bibliographystyle{abbrv}
\bibliography{biblio}

\end{document}